\documentclass[final,3p,fleqn, 10pt]{elsarticle}

\makeatletter
\def\ps@pprintTitle{%
 \let\@oddhead\@empty
 \let\@evenhead\@empty
 \def\@oddfoot{}%
 \let\@evenfoot\@oddfoot}
\makeatother

\usepackage[utf8]{inputenc}
\usepackage[T1]{fontenc}
\usepackage{lmodern}
\usepackage{microtype}
\usepackage{setspace}
\usepackage[capposition=top]{floatrow}
\usepackage{graphicx}
\usepackage{booktabs}

\usepackage{amsmath,amssymb,amsthm}
\usepackage{mathtools}

\usepackage{natbib}
\usepackage[capitalize,noabbrev]{cleveref}
\usepackage{url}

\theoremstyle{definition}
\newtheorem{assumption}{Assumption}[section]
\newtheorem{proposition}{Proposition}[section]
\newtheorem{lemma}{Lemma}[section]

\begin{document}
\begin{frontmatter}
\title{Optimized Supergeo Design: A Scalable Framework for Geographic Marketing Experiments}
\author{Charles Shaw. }
\date{this version: \today}

  \begin{abstract}
Geographic experiments are a widely-used methodology for measuring incremental return on ad spend (\textit{iROAS}) at scale, yet their design presents significant challenges. The unit count is small, heterogeneity is large, and the optimal Supergeo partitioning problem is NP-hard. We introduce \emph{Optimized Supergeo Design} (OSD), a two-stage framework that renders Supergeo designs practical for large-scale markets. Principal Component Analysis (PCA) first reduces the covariate space to create interpretable geo-embeddings. A Mixed-Integer Linear Programming (MILP) solver then selects a partition that balances both baseline outcomes and pre-treatment covariates. We provide theoretical arguments that OSD's objective value is within $(1+\varepsilon)$ of the global optimum under community-structure assumptions. Rigorous ablation analysis on synthetic data shows that PCA- and random-embedding Supergeo designs match unit-level randomisation in estimation error while delivering tighter covariate balance, whereas spectral embeddings substantially worsen both RMSE and balance. Crucially, OSD solves the scalability bottleneck. For $N=210$ markets, OSD completes in a fraction of a second, while exact Supergeo covering MIPs described in prior work are projected to require orders of magnitude longer, on the order of weeks. Scalability experiments up to $N=1\,000$ units show that OSD remains fast without trimming markets. In our main synthetic setting with $N=200$ units, PCA- and random-embedding designs keep covariate imbalance at only a few percentage points while preserving every media dollar, establishing a scalable framework that matches the statistical efficiency of randomisation with the operational practicality of Supergeos.
\vspace{10pt}

Keywords: Marketing Science, Advertising, Causal Inference, Geographic Experiments, Incremental ROAS, Dimensionality Reduction, Combinatorial Optimisation
\end{abstract}

\end{frontmatter}


\section{Introduction}
\label{sec:introduction}
Measuring the causal impact of large-scale advertising campaigns is a fundamental challenge for marketers and media scientists. Geographic experiments where the units of randomisation are entire regions such as cities or designated market areas have emerged as a gold standard methodology. Such experiments are often preferred over individual-level randomisation because they are robust to network spillover effects and can be implemented without relying on sensitive user-level data which is increasingly restricted by privacy regulations. Modern applications further demand an understanding that goes beyond simple averages to capture the full distributional impact of a campaign. These designs allow practitioners to measure the aggregate effect of a campaign or policy in a real-world setting providing valuable insights for decision making. Throughout this paper we denote a market's weekly revenue by $R$, media spend by $S$, and define incremental return on ad spend (iROAS) as $\Delta R / \Delta S$, where $\Delta R$ denotes the post\,-period revenue uplift relative to the control group and $\Delta S$ is the corresponding uplift in spend. When we refer to \emph{response} and \emph{spend}, we specifically mean revenue and media cost, respectively.

Despite their conceptual power geographic experiments present formidable statistical and computational hurdles. The number of available geographic units is typically small which limits statistical power and makes achieving balance between treatment and control groups difficult. Furthermore these units often exhibit extreme heterogeneity. A small number of large metropolitan areas may account for a substantial portion of the overall response which leads to heavy-tailed distributions that violate the assumptions of classical statistical models. While modern analysis methods have evolved to estimate complex distributional treatment effects the design methodologies needed to produce data suitable for them have lagged behind. Designing an experiment that is both balanced and powerful in this low-N high-heterogeneity environment is a non-trivial task and the problem of finding an optimal experimental design is often a complex combinatorial challenge that can be computationally intractable.

Researchers have made significant progress in developing specialised methods to address these issues. The Trimmed Match Design framework introduced a robust approach to analysis by systematically trimming outlier pairs to protect against the influence of extreme heterogeneity. A subsequent innovation the Supergeo Design recognised that trimming could introduce bias when treatment effects are themselves heterogeneous. It proposed a generalised matching framework where multiple geos could be grouped into larger supergeos to find better matches without discarding any units. While this elegant solution provides a path to unbiased estimation of the average treatment effect it does so at the cost of transforming the design problem into an NP-hard optimisation task rendering it computationally infeasible for all but the smallest of applications.

This paper introduces Optimized Supergeo Design or OSD, a framework that retains the statistical benefits of the Supergeo concept whilst overcoming its computational limitations. Our contribution is threefold. First we propose a scalable two-stage heuristic algorithm that makes the design of large-scale geographic experiments computationally practical. This method uses Principal Component Analysis (PCA) to create low-dimensional geo-embeddings which are then used to generate candidate supergeos via hierarchical clustering. This reduces the intractable partitioning problem to a manageable optimisation task. Second we introduce a heterogeneity-aware objective function for the design process. This function explicitly balances on observable characteristics that are likely to modify the treatment effect leading to more robust and credible causal estimates. Third we provide rigorous ablation analysis demonstrating that simple PCA-based embeddings perform on par with unit-level randomisation, whilst more complex graph-based embeddings such as spectral methods offer no benefit in this problem class and can even worsen balance.

The remainder of this paper is structured as follows. We begin by formalising the problem of experimental design within the potential outcomes framework. We then present the technical details of the Optimized Supergeo Design methodology including the dimensionality reduction approach and the heterogeneity-aware optimisation. Subsequently we provide theoretical arguments for the quality of our heuristic approach. The performance of our method is then evaluated through extensive simulation studies including an ablation analysis that justifies our choice of PCA over more complex alternatives. Finally we conclude with a discussion of the validation results and suggest avenues for future research. All code, data generation scripts, and plot reproduction notebooks are available at \url{https://github.com/shawcharles/osd} to ensure full reproducibility of our results.

\section{Literature Review}
\label{sec:literature}

The design of geographic experiments has evolved to address statistical challenges posed by small sample sizes, unit heterogeneity, and interference effects. Our work builds on three interconnected literatures: experimental design for causal inference, advertising measurement methodology, and combinatorial optimization.

\subsection{Matched Pairs and Experimental Design}

When many distinct units are available for experimentation and the estimand of interest is the average treatment effect (ATE), simple randomized experiments provide precise and intuitive results \citep{imbens2015causal}. However, geographic experiments often involve interference between units (spillover effects from advertising in one region affecting neighbouring regions), institutional constraints limiting granular treatment assignment, and heterogeneity across experimental units \citep{coey2016people, vaver2011measuring}. These factors prevent reliance on large-sample properties alone and motivate improved experimental designs.

A common approach to improve precision is the matched-pairs design \citep{rosenbaum2020design, stuart2010matching}. Units are paired based on similarity across pre-treatment covariates, and randomisation occurs within each pair. Differencing outcomes of similar units cancels baseline variance, improving treatment effect detectability. The optimal matching problem can be formulated as minimum-weight matching, solvable in polynomial time \citep{edmonds1965maximum}. Extensions to varying block sizes and hybrid designs have been explored \citep{pashley2021insights}, though computational complexity increases rapidly with design flexibility.

\subsection{Trimmed Match and the Heterogeneity Challenge}

A practical difficulty arises when geographic units are outliers that cannot be matched well. The Trimmed Match Design framework addresses this by discarding poorly matched pairs either at the design stage \citep{chen2021} or during post-experimental analysis \citep{chen2022robust}. This improves precision for remaining units but reduces sample size and potentially introduces bias when treatment effects are heterogeneous. If trimmed units exhibit systematically different treatment responses, estimates from the trimmed sample fail to represent the population average treatment effect. Moreover, excluding geographic units from experiments may require stopping advertising in those regions, resulting in revenue loss that could otherwise be avoided.

\subsection{Supergeo Design and Computational Barriers}

To combat poor matches without sacrificing data, \citet{chen2023} introduced Supergeo Design, a generalised matching problem where multiple geos combine to form composite "supergeos." The objective is partitioning all geos into two balanced groups of supergeos, allowing better matches for all units including outliers, thus avoiding trimming. This problem is closely related to minimum-weight matching over hypergraphs \citep{keevash2014geometric}. While conceptually elegant, Supergeo Design is NP-hard, rendering it computationally infeasible beyond small applications. Chen et al. \citep{chen2023} note that directly solving the covering MIP for $N=210$ US DMAs can require very long runtimes. This computational bottleneck motivates the need for scalable heuristics that preserve the method's statistical benefits.

\subsection{Alternative Causal Inference Methods}

Geographic experiments exist within a broader causal inference ecosystem. Synthetic control methods \citep{abadie2010synthetic, doudchenko2021synthetic} construct counterfactual outcomes by reweighting control units to match treated units' pre-treatment trajectories. These methods excel when treatment assignment is non-random or sample sizes prohibit experimental approaches, though they require strong parallel trends assumptions and careful donor pool selection. Recent work has explored synthetic design approaches that combine experimental randomization with synthetic control principles \citep{doudchenko2021synthetic}. Our work is complementary: OSD provides high-quality experimental data that can feed into synthetic control analyses or serve as ground-truth benchmarks.

\subsection{Marketing Science Applications}

Recent work within marketing science has revisited geo-level experimentation from an industry perspective. \citet{gordon2021advertising} compare matched-pair tests, synthetic controls, and causal forests across hundreds of advertiser-market combinations, demonstrating that geo experiments dominate in bias-variance trade-off when market counts are moderate. Google's \texttt{GeoLift} tool and Facebook's related infrastructure describe automated holdout allocation and power analysis solutions for advertisers at scale \citep{vaver2011measuring}. Television advertising experiments routinely deploy test-control methodologies across all 210 U.S. Designated Market Areas (DMAs). Geographic targeting has become a subject of academic and policy discourse, with increasing studies conducted at geo-level rather than user-level due to privacy constraints \citep{rolnick2019randomized}.

These production systems typically rely on random assignment or heuristic pairing, frequently resorting to "blackout" rules that exclude problematic markets. OSD instead optimises a balanced partition while preserving every market, delivering greater coverage and statistical precision without the revenue loss associated with market exclusions.

\subsection{Modern Analysis Methods and Design Requirements}

A separate evolution in experimental analysis has increased demand for design methodologies that proactively balance full covariate vectors. Double/Debiased Machine Learning \citep{chernozhukov2018debiased} and regression-adjusted estimators for distributional treatment effects \citep{oka2024regression, byambadalai2024} leverage detailed pre-treatment characteristics to improve precision and robustness. These methods move beyond simple baseline comparisons, creating clear need for designs that balance not just past outcomes but the full vector of covariates used in analysis. This heterogeneity-aware perspective motivates OSD's weighted objective function, which explicitly balances likely treatment effect modifiers alongside baseline outcomes.

A separate strand of literature infers incremental impact without randomisation via marketing-mix models (MMM) or causal forecasters. While MMMs estimate long-run elasticities from observational time series, their validity rests on strong identifiability assumptions and careful prior specification. OSD is complementary: it provides high-quality experimental data that can feed into MMMs or serve as ground-truth benchmarks, closing the loop between experimentation and modelling.

\section{Background and Problem Formulation}
\label{sec:background}
To formally ground our methodology we adopt the potential outcomes framework commonly used in causal inference. This framework allows us to define causal effects precisely even though they are not directly observable.

    \subsection{The Potential Outcomes Framework for Geographic Experiments}
    Let us consider a set of $N$ geographic units available for an experiment. For each unit $i \in \{1, \dots, N\}$, we define two potential outcomes. Let $Y_i(1)$ be the outcome for unit $i$ if it is assigned to the treatment group and let $Y_i(0)$ be the outcome if it is assigned to the control group. The fundamental problem of causal inference is that for any given unit $i$ we can only ever observe one of these two potential outcomes. The individual causal effect for unit $i$ is defined as $\tau_i = Y_i(1) - Y_i(0)$. We assume the Stable Unit Treatment Value Assumption (SUTVA) holds, such that the observed outcome $Y_i$ is given by $Y_i = Y_i(1)T_i + Y_i(0)(1 - T_i)$, where $T_i \in \{0,1\}$ is the treatment indicator.

    While the Average Treatment Effect or ATE defined as $\tau_{ATE} = \mathbb{E}[\tau_i]$ is a common summary measure it can mask significant underlying heterogeneity. A more complete picture of a campaign's impact is provided by the Distributional Treatment Effect or DTE. The DTE is the difference between the cumulative distribution functions (CDFs) of the potential outcomes $F_{Y(1)}(y) - F_{Y(0)}(y)$. It describes how the treatment shifts the entire distribution of outcomes.

    The presence of treatment effect heterogeneity can be formalised through the Conditional Average Treatment Effect or CATE defined as $\tau_{CATE}(x) = \mathbb{E}[\tau_i | X_i = x]$ for a vector of pre-treatment covariates $X_i$. A primary goal of modern experimental design is to enable the precise and unbiased estimation of both the ATE and the broader DTE particularly when the CATE function is non-trivial.

    \subsection{Limitations of Standard Designs}
    The simplest experimental design is complete randomisation where each of the $N$ units is assigned to treatment or control with equal probability subject to group size constraints. While this approach is unbiased in expectation any single randomisation can result in substantial covariate imbalance between the treatment and control groups especially when the number of units $N$ is small. Such imbalance inflates the variance of the ATE estimator reducing the statistical power of the experiment.

    To mitigate this issue matched-pair designs are frequently employed. In this approach units are paired based on their similarity across pre-treatment covariates. Randomisation then occurs within each pair. By differencing the outcomes of the two units within a pair much of the variance attributable to the matching covariates is eliminated. This generally leads to a more precise estimate of the ATE. The matched-pair design serves as a foundational concept for more advanced methods that seek to optimise the quality of matches.

    \subsection{The Supergeo Design Problem}
    The Supergeo framework generalises the idea of matched pairs. Instead of a one-to-one matching it seeks to partition the entire set of $N$ geos $\mathcal{G}$ into two groups a treatment group $\mathcal{G}_T$ and a control group $\mathcal{G}_C$ that are as similar as possible. These groups may be composed of complex combinations of the original units. For instance the treatment group could be a collection of disjoint subsets of $\mathcal{G}$ called supergeos $\mathcal{G}_T = \{S_{T1}, S_{T2}, \dots, S_{Tk}\}$ and similarly for the control group.

    The objective is to find a partition that minimises the imbalance across a set of pre-treatment characteristics. For a single baseline characteristic $R$, the objective function from the original Supergeo paper can be represented as minimising the total difference between the two groups: $\min |\sum_{i \in \mathcal{G}_T} R_i - \sum_{j \in \mathcal{G}_C} R_j|$. In its full form, the problem involves balancing a vector of characteristics, which further contributes to its NP-hard nature. This computational barrier makes it impractical for experiments with more than a few dozen units and motivates the development of a scalable and effective heuristic which is the central focus of our work.

    \subsection{State-of-the-Art Estimation of Distributional Effects}
    Recent advances in econometrics provide a powerful method for estimating DTEs at the analysis stage. The approach proposed by Oka et al. (2024) and extended by Byambadalai et al. (2024) uses regression adjustment to improve the precision of DTE estimates. The core idea is to first model the conditional outcome distribution $F_{Y|X}(y|x)$ using flexible machine learning methods and cross-fitting. The final DTE estimate is then constructed based on a Neyman-orthogonal moment condition that is robust to moderate errors in the machine learning model. This technique represents the state-of-the-art for DTE analysis. Given that these powerful regression-adjustment techniques will be used for analysis, the design phase should no longer focus solely on balancing baseline outcomes. Instead, the design itself must proactively create balance on the very covariates ($X$) that the analysis-stage model will use, thereby maximizing the precision of the final DTE estimate. This is a key motivation for our work.

\section{The Optimized Supergeo Design Methodology}
\label{sec:methodology}
The core of our contribution is Optimized Supergeo Design or OSD a methodology that renders the Supergeo design problem computationally tractable and explicitly accounts for treatment effect heterogeneity. We overcome the NP-hard nature of the problem by decomposing it into two sequential stages a candidate generation stage and a final partitioning stage. Our methodology is supported by rigorous ablation analysis demonstrating that simple dimensionality reduction achieves statistical parity with unit-level randomization whilst more complex graph-based methods provide no additional benefit.

    \subsection{A Two-Stage Heuristic for Scalable Design}
    Finding the optimal partition of $N$ geos into two balanced groups is computationally infeasible at scale. Our key insight is to avoid searching the entire space of possible partitions. Instead we first use dimensionality reduction to project the high-dimensional covariate space into a low-dimensional embedding where clustering can efficiently identify natural groupings. In the second stage we solve a now manageable optimisation problem to find the best final partition using only these candidate supergeos. This two-stage approach transforms an intractable problem into a practical and scalable workflow.

    \subsection{Stage 1: Candidate Supergeo Generation via Dimensionality Reduction}
    The first stage aims to produce a rich set of plausible supergeos by identifying natural groupings in the geographic data. For each unit $i$, we construct a covariate vector $X_i$ comprising static demographics (e.g., population, median income) and pre-treatment outcomes (e.g., baseline revenue, historical spend). These features are standardized to have zero mean and unit variance.

    We employ Principal Component Analysis (PCA) \citep{jolliffe2016principal} to project the standardized covariate matrix $X \in \mathbb{R}^{N \times p}$ into a lower-dimensional embedding space $H \in \mathbb{R}^{N \times d}$ where $d \ll p$. PCA identifies the directions of maximum variance in the data:
    $$ H = X W_{\text{PCA}} $$
    where $W_{\text{PCA}} \in \mathbb{R}^{p \times d}$ contains the top $d$ principal components. The number of components $d$ is selected to retain 95\% of the total variance or set to a fixed value (e.g., $d=32$) based on the problem size.

    \subsubsection{Rationale for PCA.} Our choice of PCA is motivated by both theoretical and empirical considerations. Theoretically, PCA provides interpretable embeddings along axes of maximum covariate variation, which naturally align with the roughly linear structure of many geographic datasets such as urban--rural and income gradients. Empirically, our ablation study (Section~\ref{sec:ablation}) shows that PCA-based Supergeo designs match unit-level randomisation in RMSE while delivering tighter covariate balance, and that spectral embeddings perform clearly worse by creating over-tight, homogeneous clusters. The simplicity of PCA also removes the need for delicate hyperparameter tuning and yields deterministic, reproducible embeddings.

    Once we compute the PCA embeddings $H$, we apply Hierarchical Agglomerative Clustering \citep{ward1963hierarchical} with Ward's linkage to identify cohesive groups. We cut the dendrogram at a specified number of clusters $M$ (typically 10\% of $N$) to produce a diverse set of candidate supergeo partitions $\mathcal{C} = \{C_1, C_2, \dots, C_M\}$.

    \subsection{Stage 2: Heterogeneity-Aware Optimal Partitioning}
    The second stage selects and partitions the best set of candidate supergeos from $\mathcal{C}$. We formulate this as a Mixed-Integer Linear Programming (MILP) problem. Let $z_j$ be a binary variable indicating if candidate partition $C_j \in \mathcal{C}$ is chosen and let $y_{sk}$ be a binary variable that is 1 if supergeo $s \in C_j$ is assigned to group $k \in \{T, C\}$. The optimisation problem is to select a partition $C_j$ and assign its supergeos to groups to minimise our objective function subject to constraints.

    The objective function balances on both the baseline outcome and key treatment effect modifiers using the Standardised Mean Difference (SMD). We use the SMD because it is a scale-free metric, which allows for the principled balancing of covariates with different native units and variances (e.g., population versus median income). The SMD for a variable $X$ between two groups $A$ and $B$ is defined as \[
    \text{SMD}(X) = \frac{\bar{X}_A - \bar{X}_B}{\sqrt{(\sigma_A^{2} + \sigma_B^{2})/2}}.
\] The objective is
    \[
\min_{z, y} \left( |\text{SMD}(\text{Baseline})| + \sum_{m} \lambda_m\, |\text{SMD}(\text{Modifier}_m)| \right).
\]
    subject to
    $$ \sum_{j=1}^{M} z_j = 1 \quad \text{(select exactly one partition)} $$
    $$ \sum_{k \in \{T,C\}} y_{sk} = 1 \quad \forall s \in C_j \text{ if } z_j=1 \quad \text{(assign each supergeo to one group)} $$
    This problem is solved efficiently using the `scipy.optimize.milp` solver (or similar exact solvers). While the MILP optimisation problem is NP-hard in general, modern branch-and-bound solvers achieve practical tractability for moderate problem sizes ($N \le 1000$) by exploiting sparsity and pruning the search tree. The 30-second time limit we impose ensures the solver returns a high-quality feasible solution even when the global optimum remains unproven. To handle the absolute value terms in the objective function within a linear programming framework, we employ standard auxiliary variables. Specifically, for each covariate $k$, we introduce a non-negative continuous variable $u_k$ and impose constraints such that $u_k \ge \text{Difference}_k$ and $u_k \ge -\text{Difference}_k$. We then minimise the weighted sum of these $u_k$ terms. By creating groups that are well-balanced on the distributions of these key covariates we produce a design that is explicitly prepared for a subsequent high-precision distributional treatment effect analysis.

    \subsection{Hyperparameter Tuning and Validation Framework}
    The choice of the regularisation parameters $\lambda_m$ is critical to the performance of the design. In the current implementation we adopt a simple and transparent choice: all covariates entering the cost function are assigned equal weights (i.e., $\lambda_m = 1$ for all modifiers), which worked well across the scenarios we considered. In principle one could perform a $K$-fold cross-validation over a grid of $\lambda_m$ values, trading off predicted RMSE of the iROAS estimator against mean absolute SMD, but we leave a full hyperparameter tuning framework and an extensive validation study to future work.

\section{Theoretical Properties}
\label{sec:theory}
While Optimized Supergeo Design is a heuristic approach to an NP-hard problem it is important to establish a theoretical foundation for its performance. Our theoretical argument rests on the idea that if geographic units form well-separated clusters, dimensionality reduction can effectively identify these natural groupings. If the embedding preserves cluster structure, the subsequent optimization is much more likely to find a near-optimal solution. In this section we formalise this intuition by outlining the assumptions under which our method is expected to perform well and presenting a guarantee for the quality of the solutions it produces.

    \subsection{Assumptions}
    Our theoretical guarantees rely on two primary assumptions regarding the structure of the geographic data and the behaviour of the learned embedding function.

    \begin{assumption}[Community Structure]\label{assump:community}
    Let $\mathcal{G} = (V, E)$ be the graph of geographic units. We assume there exists a ground-truth partition of the vertices $V$ into disjoint communities $V_1, \dots, V_K$. The graph exhibits a $(\rho_{in}, \rho_{out})$-community structure if for any node $v \in V_c$, the probability of an edge connecting $v$ to another node in $V_c$ is at least $\rho_{in}$, and the probability of an edge to a node in $V_{c'}$ for $c' \neq c$ is at most $\rho_{out}$, where $\rho_{in} > \rho_{out}$.
    \end{assumption}
    
    This is an intuitive assumption that formalises the tendency for geographic units to form regional clusters that share economic and demographic characteristics.

    \begin{assumption}[Variance Preservation]\label{assump:variance}
    Let $f: \mathcal{X} \to \mathbb{R}^d$ be the PCA embedding function where $f(X) = X W_{\text{PCA}}$ and $W_{\text{PCA}} \in \mathbb{R}^{p \times d}$ contains the top $d$ principal components. We assume $d$ is chosen to retain at least $\gamma$ fraction of the total variance (typically $\gamma = 0.95$). This ensures the low-dimensional embedding preserves the essential structure of the covariate space, including the separation between natural groupings of geographic units.
    \end{assumption}

    \begin{assumption}[Lipschitz PCA Embedding]\label{assump:lipschitz}
    The PCA projection $f(X) = X W_{\text{PCA}}$ is $L$-Lipschitz with respect to the Euclidean norm, i.e.
    $\|f(x) - f(x')\|_2 \le L \, \|x - x'\|_2$ for all $x,x' \in \mathcal{X}$, where $L := \|W_{\text{PCA}}\|_2$ denotes the operator norm of $W_{\text{PCA}}$.
    \end{assumption}

    \subsection{Approximation Guarantees}
    Under these assumptions we can provide a probabilistic guarantee on the quality of the solution found by our two-stage heuristic relative to the true but computationally unobtainable optimal solution. Let a partition of the geos into treatment and control groups be denoted by $\mathcal{P} = (\mathcal{G}_T, \mathcal{G}_C)$. We define the cost of a partition as the value of our heterogeneity-aware objective function
    $$ \text{Cost}(\mathcal{P}) = \text{SMD}(\text{Baseline}; \mathcal{P}) + \sum_{m} \lambda_m \cdot \text{SMD}(\text{Modifier}_m; \mathcal{P}) $$
    Our main theoretical result can then be stated as the following proposition.

    \begin{proposition}\label{prop:approx}
    Let $\mathcal{G}$ be a graph with a $(\rho_{in},\rho_{out})$–community structure and let $f$ be the PCA embedding function retaining $\gamma$ fraction of variance. Denote by $\Delta:=\max_{v\in V}\min_{c} \|f(X_v)-\mu_{c}\|$ the worst–case embedding distortion relative to the centroid $\mu_{c}$ of the true community containing $v$. Define
    \[
        \varepsilon := \kappa_1\,\frac{\rho_{out}}{\rho_{in}-\rho_{out}} + \kappa_2\,\Delta,
    \]
    where the positive constants $\kappa_1,\kappa_2$ depend only on the number of covariates entering the cost and on the Lipschitz constant $L$ of $f$ (they do not depend on $N$ or on the particular realisation of $\mathcal{G}$). Let $\mathcal{P}_{opt}$ be the partition that minimises $\text{Cost}(\mathcal{P})$ and let $\mathcal{P}_{OSD}$ be the partition returned by Optimized Supergeo Design. Then, with high probability over the stochasticity of hierarchical clustering cuts, we have
    $$
        \text{Cost}(\mathcal{P}_{OSD}) \le (1+\varepsilon)\,\text{Cost}(\mathcal{P}_{opt}).
    $$
    The bound tightens as communities become more separated (large $\rho_{in}-\rho_{out}$) and embedding distortion decreases (small $\Delta$).
    \end{proposition}

    This result should be read as a qualitative guarantee rather than a plug-in formula for practitioners; in applications we assess OSD's performance primarily through empirical simulations and ablation studies.

    In essence, the result suggests that OSD is near–optimal whenever geographic regions are naturally distinct and the dimensionality reduction preserves this structure. Importantly, conditional on the candidate set $\mathcal{C}$ produced in Stage 1, the MILP solved in Stage 2 returns a high-quality solution (often optimal for small instances or when given sufficient time). Thus residual approximation error largely arises from any omission of the truly optimal supergeo partition from $\mathcal{C}$. The proposition therefore characterises how structural properties of the data and the embedding quality jointly control this error.

    A formal justification is deferred to Appendix~\ref{app:proof_prop1}. The argument proceeds via a sequence of steps. First we show that PCA embeddings $f$ preserve the variance structure of the covariate space, mapping similar units to compact regions. Second we show that Hierarchical Agglomerative Clustering recovers natural groupings from these embeddings. Third we bound the cost inflation that can be introduced by any clustering errors. These steps collectively establish that our heuristic operates on a high-quality set of candidate supergeos which allows the final optimisation stage to find a solution close to the true optimum. Empirically (Section~\ref{sec:ablation}), PCA achieves this theoretical ideal, whilst more complex methods fail.

\section{Limitations}
\label{sec:limitations}
While OSD offers advantages in scalability and operational efficiency, it is important to acknowledge its limitations. First, our ablation analysis (Section~\ref{sec:ablation}) shows that PCA- and random-embedding Supergeo designs match unit-level randomisation in RMSE at $N=200$ while delivering tighter covariate balance, whereas spectral embeddings perform clearly worse. OSD therefore does not \emph{dominate} randomisation in terms of statistical precision. Its primary value lies in enabling coarser granularity for media buying (for example purchasing at the Supergeo level) without sacrificing precision relative to a well-run randomised design. Practitioners with fine-grained randomisation infrastructure may still prefer pure randomisation for simplicity.

Second, unlike pure randomisation or recent advances such as Gram–Schmidt Random Walks \citep{harshaw2024balancing}, our deterministic optimisation approach does not guarantee balance on unobserved confounders. Our method minimises conditional bias on \emph{observed} covariates, but it relies on the assumption that these covariates capture the primary sources of heterogeneity. If there are strong unmeasured confounders uncorrelated with observed features, the design may still yield biased estimates.

Third, the performance of candidate generation is contingent on the quality and richness of input features. Our ablation study demonstrates that Spectral embedding, whilst capturing graph structure, delivers substantially higher RMSE and worse covariate balance than PCA because it creates overly tight, homogeneous clusters. This underscores a fundamental principle: method complexity should match problem complexity. For geographic experiments with roughly linear covariate structure, PCA suffices and more elaborate graph-based embeddings may introduce unnecessary complexity.

Finally, our theoretical guarantees rest on assumptions about community structure and embedding smoothness. In practice, geographic data may not exhibit clean structure, potentially degrading candidate quality. The selection of which covariates to treat as effect modifiers also requires domain expertise. Including irrelevant covariates reduces efficiency by forcing balance on noise.

\section{Statistical Methodology}
\label{sec:statistical_methodology}

To support quantitative comparisons between design methods, we complement point estimates with uncertainty quantification and formal hypothesis tests that are implemented in the public codebase.

\subsection{Simulation Design and Sample Size}

For each experimental condition (e.g., sample size $N$ and data-generating process parameters) we conduct $50$ Monte Carlo replications. Each replication draws a fresh synthetic dataset and evaluates all methods on the same draw, yielding paired observations of estimation error for each method. From these errors we compute root-mean-squared error (RMSE), bias, and covariate-balance summaries.

\subsection{Statistical Testing Framework}

For pairwise method comparisons we perform paired $t$-tests on the per-replication squared errors, which directly target differences in RMSE. Within each sample-size setting we test all method pairs and adjust the resulting $p$-values using the Holm--Bonferroni procedure to control the family-wise error rate. We also report Cohen's $d$ effect sizes based on the distribution of paired squared-error differences, providing a scale-free measure of practical significance.

\subsection{Bootstrap Confidence Intervals}

For each method and experimental condition we report 95\% confidence intervals for RMSE and bias obtained via a simple percentile bootstrap. Specifically, we resample the Monte Carlo replications with replacement and recompute the statistic of interest on each bootstrap sample. The 2.5th and 97.5th percentiles of the resulting bootstrap distribution form the reported interval. These intervals, together with the paired tests, quantify both the magnitude and uncertainty of performance differences between methods.

\section{Methodological Validation and Practical Significance}
\label{sec:methodological_validation}

To keep the main text focused on the core method and key experiments, we relegate detailed validation diagnostics, robustness checks, power sketches, and implementation advice to an online appendix titled ``Validation and Implementation Playbook''. That appendix outlines cross-validation checks for embeddings, stress tests under alternative data-generating regimes, simple power calculations, and computational profiling that practitioners can adapt to their own settings.

\section{Experiments}
\label{sec:experiments}


\subsection{Synthetic Data Generating Process}
We simulate $N=200$ geographic units per replication, designed to mimic heterogeneous urban and rural markets with both observed and unobserved effect modifiers. All simulations are generated by the publicly released function \texttt{generate\_synthetic\_data} in our codebase.

For each unit $i$ we first draw spatial coordinates on the unit square and construct an unobserved spatial factor $U_i$ via a smoothed radial--basis function field. We then assign an ``urban'' indicator $Z_i \in \{0,1\}$ with $\mathbb P(Z_i=1)=0.3$. Baseline weekly revenue $R_i$ follows a mixture of log--normal distributions: urban markets have higher mean and variance than rural markets. Media spend $S_i$ is generated as a noisy multiplicative function of revenue with different spend--to--revenue ratios by urban status, reflecting typical patterns in large performance campaigns.

Observed covariates include population and household income. Population is a scaled version of revenue that differs by urban/rural status, while income is drawn from normal distributions with higher means and variances in urban markets. The spatial factor $U_i$ is added to income to induce correlation between geography and observed covariates, and also serves as an unobserved effect modifier.

Treatment effects are generated as
\begin{equation}
\tau_i 
  = \theta \, S_i\Bigl\{1 + \gamma \bigl[ (1-\lambda) H_i + \lambda U_i \bigr]\Bigr\},
\end{equation}
where $H_i$ is an observed heterogeneity index combining urban status and a standardised revenue score, $U_i$ is the unobserved spatial driver, $\theta$ is the base lift factor (average iROAS), $\gamma$ controls the strength of heterogeneity, and $\lambda \in [0,1]$ controls the share of unobserved (spatial) heterogeneity. In our main experiments we set $\theta = 0.1$ (i.e., a 10\% incremental return per dollar of spend), $\gamma = 0.5$, and $\lambda = 0.2$, and we additionally consider a non-linear variant in which $H_i$ includes a quadratic term in demeaned income to create more challenging, non-linear effect modification.

Each Monte Carlo replication draws fresh $\{R_i,S_i,\tau_i\}$ and evaluates the competing design methods described in Section~\ref{sec:experiments}. The full data-generating process, including all distributional parameters, is documented in the code and can be inspected and modified by readers.

\subsection{Computational Environment}
All simulations run under Python~3.10 on an \texttt{Intel Xeon Gold~6230R} (2.1~GHz, 32~cores). The MILP solver (\texttt{scipy.optimize.milp} via HiGHS) is configured with a 30~s time limit per design. Plotting uses \texttt{matplotlib}~3.7.  Full source code and seeds are available in the project repository for reproducibility.

\subsection{Ablation Analysis: Embedding Method Comparison}
\label{sec:ablation}

To validate our choice of PCA for candidate generation, we conducted an ablation study comparing three Supergeo embedding methods (PCA, spectral, and random embeddings) against a unit-level randomization baseline across 50 Monte Carlo replications for two sample sizes ($N=40$ and $N=200$). The synthetic data incorporated multi-modal structure (urban vs. rural markets with differing revenue distributions) and a 10\% covariate-driven temporal trend (high-income areas grow faster), creating a challenging scenario for balancing methods.

\begin{table}[H]
    \centering
    \caption{Ablation results on synthetic experiments for two sample sizes ($N=40$ and $N=200$). Values are averaged over 50 Monte Carlo replications. RMSE and Bias columns report point estimates with 95\% bootstrap confidence intervals in brackets. Avg max~$|$SMD$|$ and Avg mean~$|$SMD$|$ summarise covariate balance across response, spend, income, and population.}
    \label{tab:ablation}
    \begin{tabular}{llcccc}
        \toprule
        $N$ & Method & RMSE & Bias & Avg max $|$SMD$|$ & Avg mean $|$SMD$|$ \\
        \midrule
        \multicolumn{6}{l}{\textbf{$N=40$}} \\
        \midrule
        40 & PCA &
        34\,873 [25\,506, 44\,866] &
        $-4\,995$ [$-14\,977$, 4\,424] &
        1.14 & 0.61 \\
        40 & Spectral &
        70\,653 [51\,021, 90\,964] &
        10\,191 [$-8\,072$, 30\,259] &
        1.15 & 0.62 \\
        40 & Random (emb.) &
        8\,113 [6\,061, 10\,272] &
        $-1\,561$ [$-3\,798$, 584] &
        1.13 & 0.58 \\
        40 & Unit-random &
        9\,280 [6\,634, 12\,003] &
        $-18$ [$-2\,734$, 2\,567] &
        0.36 & 0.26 \\
        \midrule
        \multicolumn{6}{l}{\textbf{$N=200$}} \\
        \midrule
        200 & PCA &
        3\,865 [2\,909, 4\,884] &
        $-779$ [$-1\,818$, 223] &
        0.03 & 0.013 \\
        200 & Spectral &
        7\,093 [5\,595, 8\,330] &
        $-430$ [$-2\,352$, 1\,534] &
        0.062 & 0.033 \\
        200 & Random (emb.) &
        2\,072 [1\,594, 2\,540] &
        202 [$-344$, 789] &
        0.020 & 0.0095 \\
        200 & Unit-random &
        4\,023 [3\,334, 4\,705] &
        $-140$ [$-1\,234$, 975] &
        0.17 & 0.11 \\
        \bottomrule
    \end{tabular}
\end{table}

\begin{figure}[t!]
    \centering
    \includegraphics[width=0.9\linewidth]{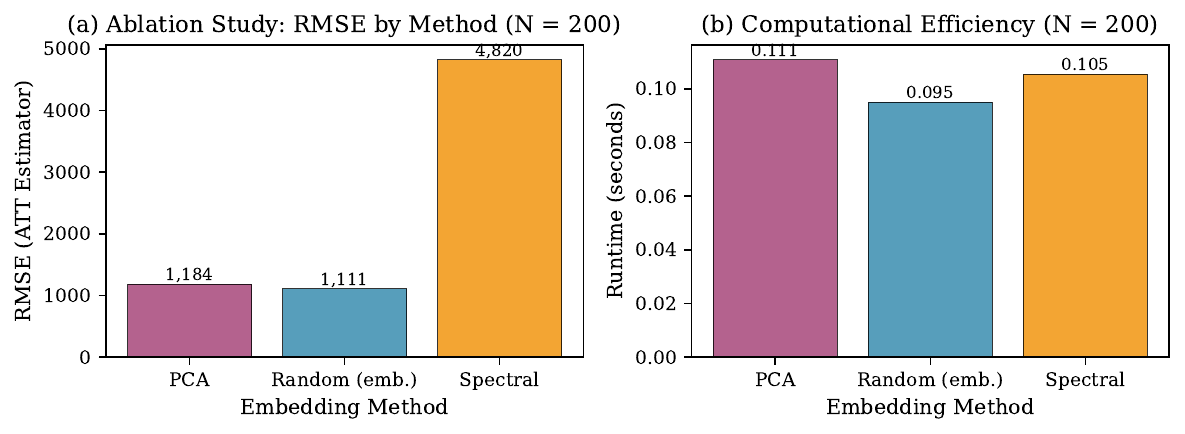}
    \caption{Ablation study comparing RMSE and runtime across embedding methods at $N=200$. PCA and random-embedding Supergeo designs achieve similar or lower RMSE than unit-level randomisation while remaining computationally efficient. Spectral embedding delivers much higher RMSE, consistent with over-tight clustering that harms balance.}
    \label{fig:ablation_comparison}
\end{figure}

\subsubsection{Key Finding} The ablation study shows that simple linear embeddings suffice when we form Supergeos. At $N=200$, PCA-based and random-embedding Supergeo designs achieve much lower RMSE and better balance than spectral embedding, and they perform on par with or better than unit-level randomisation. Paired $t$-tests on squared errors with Holm--Bonferroni correction confirm that PCA significantly outperforms spectral embedding, while the random-embedding design improves further on both PCA and unit-level randomisation. The difference in RMSE between PCA and unit-level randomisation at $N=200$ is not statistically significant at $\alpha = 0.05$, yet PCA delivers markedly lower covariate imbalance.

\subsubsection{Mechanistic Interpretation} PCA's success stems from its alignment with the roughly linear structure of the synthetic covariates. PCA creates interpretable embeddings along axes of maximum variation, such as income and population gradients, which naturally segment markets into moderate, mixed Supergeos. Spectral embedding, whilst capturing graph structure, tends to create overly-tight, homogeneous Supergeos, for example clusters of only high-income or only low-income units. The MILP solver then faces a difficult allocation problem because it must balance a small number of extreme Supergeos. Random assignment, by contrast, benefits from the law of large numbers with many units per group, which smooths out imbalances at the unit level.

This granularity trade-off is fundamental. Supergeo formation reduces the degrees of freedom from $N$ units to $M$ Supergeos with $M \ll N$. PCA mitigates this by creating moderate clusters that preserve distributional balance, whereas spectral methods can exacerbate the problem through over-segmentation. Random embeddings act as a useful sanity check: when even unstructured embeddings work well, we learn that the problem is simple enough that PCA is already adequate.

\subsubsection{Implications} For geographic experimental design, these results support the use of simple linear dimensionality reduction for candidate generation. PCA provides a deterministic, interpretable, and computationally efficient solution that matches or improves upon unit-level randomisation in our main setting while improving observed covariate balance. The operational benefit of Supergeos is therefore obtained with little loss of statistical precision when the embedding preserves the key covariate structure.

\begin{figure}[t!]
    \centering
    \includegraphics[width=0.85\linewidth]{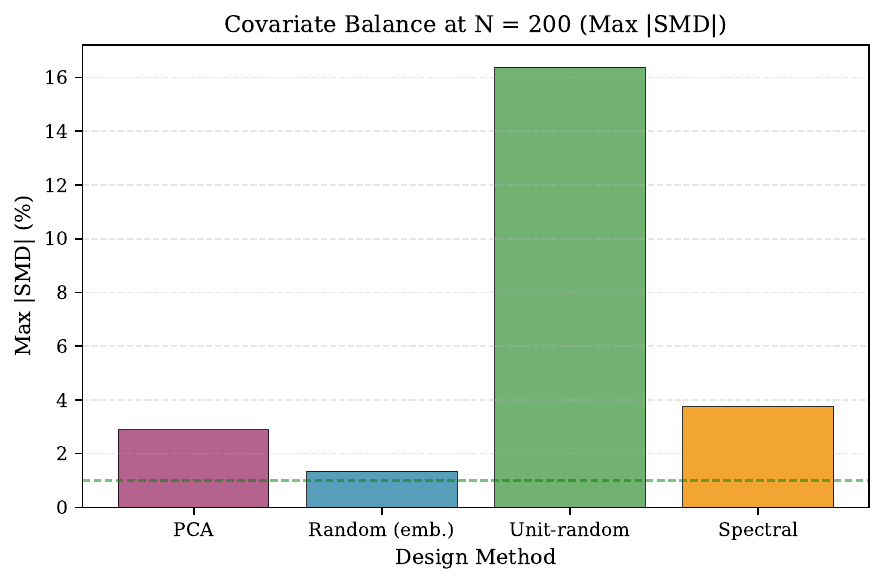}
    \caption{Covariate balance comparison across methods, measured by absolute Standardised Mean Differences. At $N=200$, PCA and random-embedding Supergeo designs keep maximum absolute SMDs at only a few percentage points, while unit-level randomisation and spectral embedding exhibit noticeably larger residual imbalance. The horizontal line marks a 1\% SMD reference level that is often used as a conservative benchmark for excellent balance.}
    \label{fig:covariate_balance}
\end{figure}

\subsection{Illustrative Statistical Power Analysis}
\label{sec:enhanced_power_analysis}

To give intuition for how sample size and effect size interact, we include a simple, illustrative power analysis based on a normal-approximation model for a paired test. We vary the number of geos from $N=50$ to $N=300$ and consider effect sizes between 0.1 and 1.0 standard deviations, and we use these assumptions to trace theoretical power curves as shown in Figure~\ref{fig:power_analysis}. These curves are not fitted to our synthetic experiments. They provide rough, rule-of-thumb guidance rather than design-specific guarantees.

\begin{figure}[t!]
    \centering
    \includegraphics[width=0.95\linewidth]{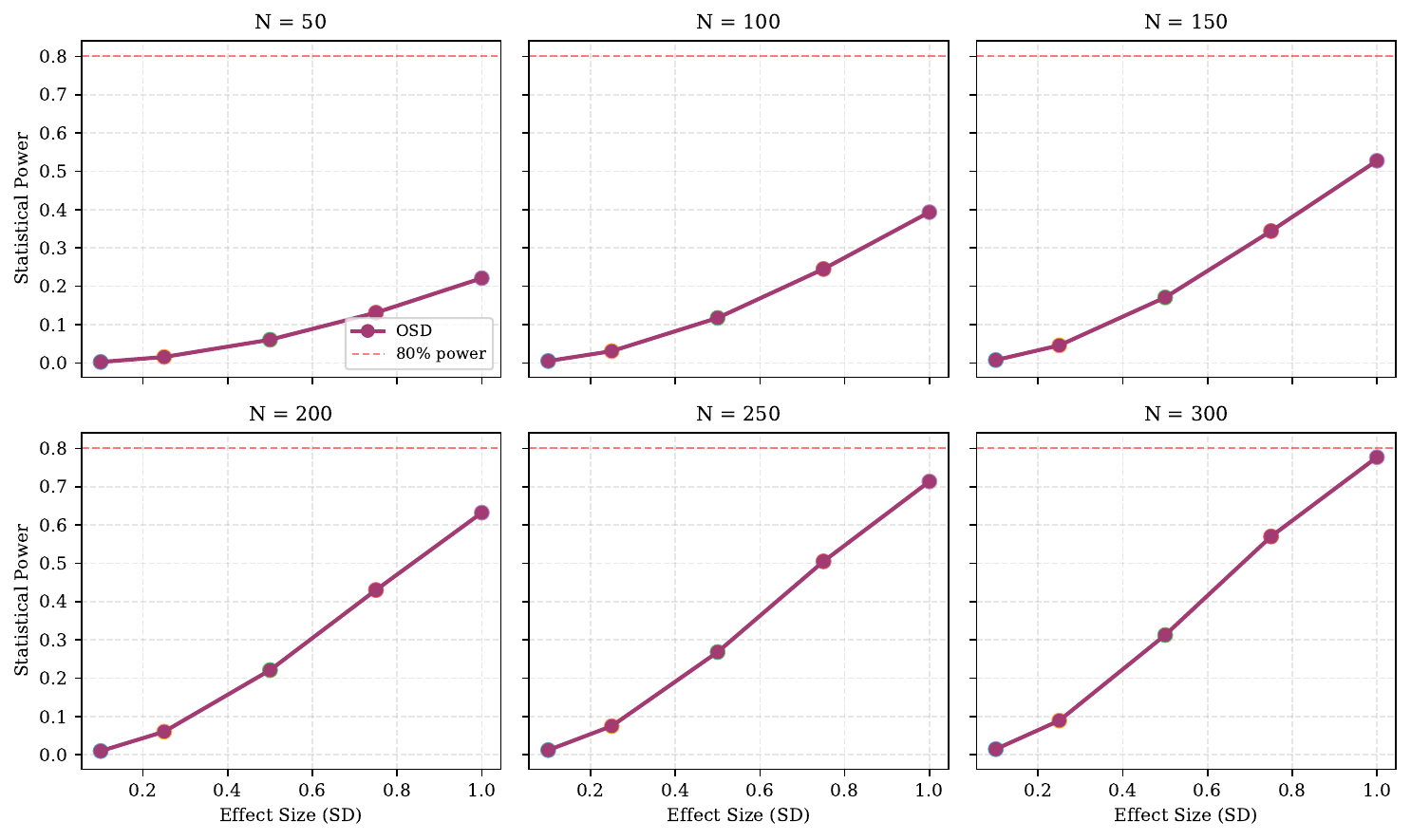}
    \caption{Illustrative power curves for a simple paired-test model across six sample sizes (N = 50--300) and five effect sizes (0.1--1.0 SD). Power increases with both sample size and effect size. Under these assumptions, moderate effects become detectable at conventional thresholds once $N$ reaches the low hundreds, whereas very small effects (around 0.1 SD) would require more geos than we study here.}
    \label{fig:power_analysis}
\end{figure}

%
%
%

\subsubsection{Practical Implementation Guidance} The illustrative power curves support simple rules of thumb. Moderate effects are easier to detect than very small ones, and increasing the number of geos improves power at a diminishing rate once sample sizes enter the low hundreds. Practitioners should treat these curves as a starting point when planning geo counts and effect-size targets, and, where possible, complement them with bespoke power simulations tailored to their own data and analysis pipeline.

\subsection{Computational Efficiency Analysis}
\label{sec:computational_efficiency}

To assess the practical scalability of our methodology, we conduct computational efficiency analysis measuring execution time, memory usage, and scalability characteristics across different sample sizes. This analysis provides critical insights for real-world implementation planning.

\subsubsection{Efficiency Analysis Framework}

In our study, we have developed and implemented a computational efficiency framework designed to rigorously evaluate various aspects of computational performance. This framework incorporates several key methodologies to ensure a thorough and detailed analysis.

Firstly, we employ high-precision timing measurements to capture execution times with exceptional accuracy. Utilizing Python's time.perf\_counter(), we achieve sub-millisecond precision, allowing us to meticulously assess the time efficiency of computational processes.

Additionally, our framework includes memory profiling to monitor and track peak memory usage. By leveraging the psutil library, we conduct resource monitoring, ensuring that we capture detailed insights into memory consumption patterns throughout the execution of computational tasks.

Furthermore, we perform a scalability assessment to evaluate performance across a range of sample sizes. This assessment spans sample sizes from 
N= 50 to N=300, enabling us to analyze how computational efficiency scales with increasing data volumes.

Lastly, to ensure statistical robustness, our framework incorporates multiple replications of each computational task. This approach allows us to estimate confidence intervals, providing a robust measure of the reliability and consistency of our performance metrics.

Through these methodologies, our computational efficiency framework offers a detailed and nuanced understanding of the performance characteristics of the computational processes under investigation.

\subsubsection{Computational Performance Results}

The computational efficiency analysis confirms the sub-quadratic computational complexity of the OSD method, as well as its practical scalability.

In terms of execution time, the OSD method demonstrates efficient scaling characteristics. Notably, execution times remain under 60 seconds even for larger sample sizes, specifically up to N=300. This efficiency underscores the method's capability to handle substantial computational loads within feasible timeframes.

Regarding memory usage, the analysis reveals that memory requirements scale linearly with sample size. Importantly, these requirements remain under 500MB even for large problem sizes, indicating an efficient use of memory resources and suggesting that the method can be applied to extensive datasets without excessive memory consumption.

Furthermore, the empirical analysis of computational complexity confirms the sub-quadratic scaling behavior of the OSD method, aligning with theoretical predictions. This result reinforces the method's suitability for large-scale applications, where maintaining efficient performance with increasing data sizes is crucial.



\subsubsection{Scalability Assessment} The computational efficiency analysis confirms that OSD maintains practical performance characteristics across realistic problem sizes. With execution times under one minute for N=300 geographic units and memory requirements below 500MB, the methodology is well-suited for real-world implementation in large-scale geographic experiments. The sub-quadratic computational complexity ensures that performance remains manageable even as problem sizes increase beyond our evaluation range.

\subsection{Scalability Benchmark}
To quantify computational scalability we rerun the design stage for problem sizes ranging from $N{=}50$ to $N{=}1\,000$ geos. Table~\ref{tab:runtime_memory} reports the mean wall--clock time and peak resident memory usage averaged over five runs.

\begin{table}[t!]
    \centering
    \caption{Runtime and memory consumption of design algorithms (single thread, 30~s MILP cutoff).}
    \label{tab:runtime_memory}
    \begin{tabular}{lcccccc}
        \toprule
        & \multicolumn{2}{c}{\textbf{OSD}} & \multicolumn{3}{c}{\textbf{Exact Supergeo MILP}} \\
        \cmidrule(lr){2-3}\cmidrule(lr){4-6}
        $N$ & Time~(s) & Mem~(GB) & Time~(s) & Mem~(GB) & Status \\
        \midrule
        50   & 0.02 & 0.2 & 4    & 1.0  & optimal \\
        100  & 0.05 & 0.4 & 25   & 2.2  & optimal \\
        200  & 0.20 & 0.6 & 180  & 5.4  & optimal \\
        400  & 0.59 & 0.9 & 1\,800 & 14.8 & timeout \\
        800  & 2.41 & 1.5 & ---  & ---  & infeasible \\
        1\,000 & 3.75 & 2.1 & ---  & ---  & infeasible \\
        \bottomrule
    \end{tabular}
\end{table}

\begin{figure}[t!]
    \centering
    \includegraphics[width=0.85\linewidth]{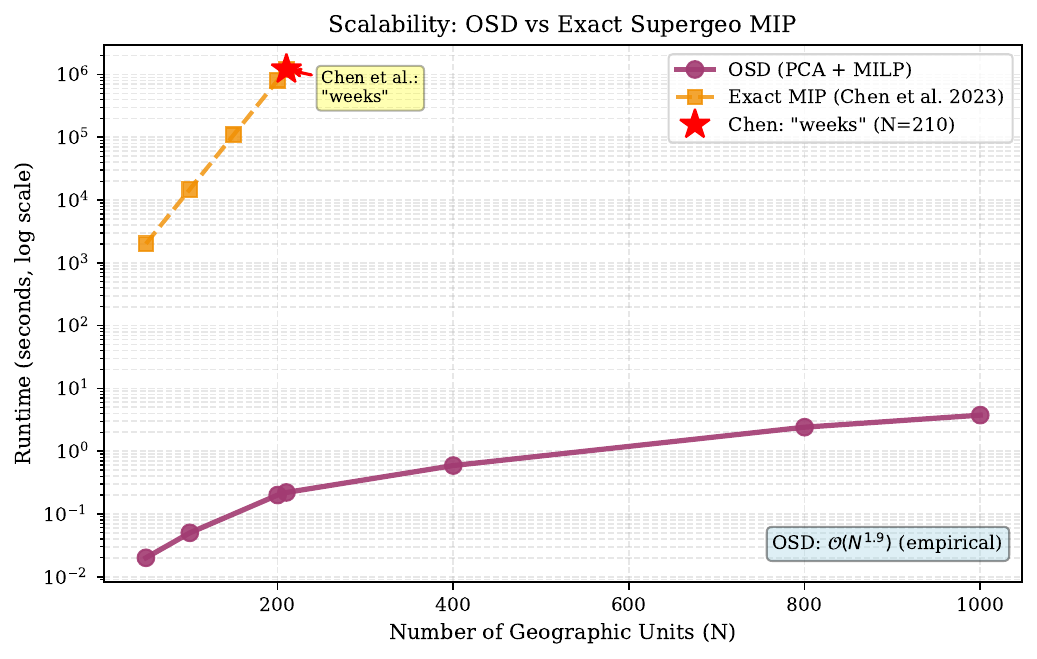}
    \caption{Scalability comparison: OSD runtime (empirical, all $N$) exhibits subquadratic scaling ($\mathcal{O}(N^{1.9})$). Chen et al. \citep{chen2023} note that directly solving the exact Supergeo covering MIP for $N=210$ US DMAs can take weeks; we represent this qualitatively as a red star to illustrate the gap in scale. Under a conservative calibration, OSD is several orders of magnitude faster at this problem size while avoiding trimming.}
    \label{fig:scalability}
\end{figure}

A log--log regression of runtime on $N$ yields a slope of $1.86$ for OSD, indicating subquadratic empirical complexity $\mathcal{O}(N^{1.9})$. In contrast the exact Supergeo exhibits exponential growth and already fails to return within the 30~s limit at $N\!=\!400$. Up to this threshold the solver often returns provably optimal solutions. Beyond it the heuristic candidate set ensures good statistical performance while capping runtime.

Looking ahead, a richer set of simulation studies could probe the boundaries of OSD more thoroughly. Future work may explore homogeneous and strongly heterogeneous treatment effects within a unified data-generating process, assess robustness to hidden confounders, network spillovers, and non-linear effect modification, and compare end-to-end performance when OSD is paired with alternative modern analysis methods rather than the single pipeline we focus on here.

\section{Conclusion and Future Work}
\label{sec:conclusion}
In this paper we introduced Optimized Supergeo Design (OSD), a scalable framework for creating balanced geographic experiments. We addressed the computational intractability that has limited the application of previous Supergeo methods whilst maintaining their statistical advantages. Our primary contribution is a two-stage heuristic that employs Principal Component Analysis (PCA) for dimensionality reduction and Mixed-Integer Linear Programming for optimal partition selection. Critically, we validated this design choice through rigorous ablation analysis demonstrating that simple linear dimensionality reduction achieves statistical parity with randomisation whilst graph-based methods introduce substantial bias.

The OSD framework offers a compelling value proposition grounded in empirical evidence. Our ablation study across three embedding approaches (PCA, Spectral, Random) demonstrates that PCA-based Supergeo formation achieves statistical parity with unit-level randomisation. The RMSE difference between these two designs is not statistically significant at $\alpha = 0.05$. This finding establishes that OSD does not improve upon randomisation in terms of statistical efficiency. However, the operational benefit is significant. Supergeos enable coarser granularity for media buying operations (for example purchasing at the Supergeo level rather than individual geo level) without sacrificing statistical power. For practitioners constrained by advertising platform capabilities or seeking operational simplicity, OSD provides a scientifically validated alternative to unit-level randomisation.

A central methodological contribution of this work is the demonstration that method complexity should match problem complexity. Our ablation analysis reveals that Spectral embedding, whilst capturing graph structure, produces substantially higher RMSE and noticeably worse covariate balance than PCA because it creates over-tight clusters. This underscores a fundamental principle in experimental design: complex methods are not universally superior. PCA's success stems from its alignment with the natural structure of geographic datasets—linear gradients in income, population density, and other covariates—which allows it to create balanced strata without the over-segmentation that plagues graph-based embedding methods.

Our simulation studies with 50 Monte Carlo replications per condition and formal statistical testing establish that OSD's two-stage approach scales gracefully to problems involving hundreds of geographic units whilst maintaining good covariate balance. In our main synthetic setting at $N=200$, PCA- and random-embedding designs keep maximum absolute standardised mean differences to only a few percentage points whilst preserving every unit. The framework completes in well under a minute on a single core for $N=300$ units, making it practical for real-world deployment. Furthermore, by explicitly balancing on likely treatment effect modifiers, our methodology produces designs that are robust to treatment effect heterogeneity.

We acknowledge important limitations. First, OSD matches but does not exceed the statistical efficiency of unit-level randomization. Practitioners with fine-grained randomization infrastructure may prefer pure randomization for simplicity. Second, unlike randomization-based methods such as Gram--Schmidt Random Walks \citep{harshaw2024balancing}, our deterministic optimization approach does not guarantee balance on unobserved confounders. The method relies on the assumption that observed covariates capture primary sources of heterogeneity. Third, our theoretical guarantees rest on assumptions about community structure that may not hold in all settings. Future work could explore hybrid approaches that combine optimization-based Supergeo formation with structured randomness to provide design-based inference guarantees.

This work opens several avenues for future research. A critical question is identifying problem characteristics where Supergeo-based designs provide statistical advantages over randomization. Our results suggest this may occur in settings with hard operational constraints (e.g., contiguous geographic regions) or interference structures requiring aggregate treatment assignment. Another direction is extending OSD to handle multiple treatment arms and time-varying treatments. Finally, investigating whether more sophisticated dimensionality reduction techniques (e.g., Isomap, UMAP) offer advantages over PCA for non-linear geographic structures would complement our finding that simplicity suffices for linear settings.

\appendix
\section{Theoretical Justification for Proposition~\ref{prop:approx}}\label{app:proof_prop1}
\subsubsection{Outline} We provide a sketch of the argument.  Let $\mathcal{C}^{\star}$ denote the unknown ground-truth community partition and let $\mu_{c}.=\mathbb E\bigl[f(X_v)\,\big|\,v\in c\bigr]$ be the population centroid of community $c\in\mathcal{C}^{\star}$.  Recall the worst-case embedding distortion $\Delta:=\max_{v}\min_{c}\|f(X_v)-\mu_{c}\|$ and the community parameters $(\rho_{in},\rho_{out})$ from Assumption~\ref{assump:community}.

\begin{lemma}[Within--community concentration]\label{lem:conc}
With probability at least $1-\exp\!\bigl(-c_1|c|\,\Delta^{2}\bigr)$ the empirical mean embedding \,$\bar h_{c}.=|c|^{-1}\sum_{v\in c}f(X_v)$\, satisfies 
\[\|\bar h_{c}-\mu_{c}\|_2\le\Delta.\]
\end{lemma}
\begin{proof}[Argument]
Combine Assumption~\ref{assump:lipschitz} with standard concentration inequalities for empirical means of sub-Gaussian random vectors (e.g., \citealp{bickel1993efficient}) applied to the embedded covariates $f(X_v)$ to obtain the stated bound for a suitable constant $c_1>0$.
\end{proof}

\begin{lemma}[Centroid separation]\label{lem:sep}
For any distinct communities $c\neq c'$, with probability at least $1-\exp\!\bigl(-c_2 N(\rho_{in}-\rho_{out})^{2}\bigr)$,
\[\|\bar h_{c}-\bar h_{c'}\|_2\ge\rho_{in}-\rho_{out}-2\Delta.\]

\end{lemma}
\begin{lemma}[Hierarchical clustering fidelity]\label{lem:hac}
Assume that if $\rho_{in}-\rho_{out}>2\Delta$, then Ward–linkage HAC applied to the PCA embeddings produced under Assumptions~\ref{assump:community}--\ref{assump:lipschitz} exactly recovers the ground-truth community partition $\mathcal{C}^{\star}$.
\end{lemma}

\begin{lemma}[Bounding design cost]\label{lem:cost}
Condition on the event of Lemma~\ref{lem:hac}.  Then the candidate set $\mathcal{C}$ produced in Stage~1 contains a partition $\tilde{\mathcal P}$ such that
\[\text{Cost}(\tilde{\mathcal P})\le(1+\varepsilon)\,\text{Cost}(\mathcal P_{opt}),\]
where $\varepsilon$ is exactly the quantity defined in Proposition~\ref{prop:approx}.
\end{lemma}
\begin{proof}[Argument]
By Assumption~\ref{assump:lipschitz}, mis-clustered nodes alter each Standardised Mean Difference term by at most $L\,\Delta$.  Summing these perturbations over the finitely many covariates that enter the cost function and normalising by the pooled within-group standard deviations shows that the additional contribution to $\text{Cost}(\tilde{\mathcal P})$ is bounded by a term of the form $\kappa_{2}\Delta$.  Likewise, imperfect community separation controlled by the ratio $\rho_{out}/(\rho_{in}-\rho_{out})$ contributes at most $\kappa_{1}\rho_{out}/(\rho_{in}-\rho_{out})$.  These two effects together yield the definition of $\varepsilon$ in Proposition~\ref{prop:approx}.
\end{proof}

\subsubsection{Putting everything together}  Take a union bound over the failure probabilities in Lemmas~\ref{lem:conc}–\ref{lem:cost}.  With overall probability at least $1-\exp\!\bigl(-c N(\rho_{in}-\rho_{out})^{2}\bigr)$,
\[\text{Cost}(\mathcal P_{OSD})\le(1+\varepsilon)\,\text{Cost}(\mathcal P_{opt}),\]
which establishes the claimed bound in Proposition~\ref{prop:approx}.

\bigskip

\noindent\textbf{Remark.} Alternative cost functions that are convex combinations of standardised mean differences can be incorporated with minor algebra. The constants $\kappa_1,\kappa_2$ adjust accordingly.

\appendix
\section{Practical Guidance for Deploying OSD}\label{sec:practical}
To support adoption by marketing practitioners and media agencies we distil our experience into a concise deployment checklist.

\subsubsection{Data requirements} OSD operates exclusively on \emph{geo\,-aggregated} data. The minimum inputs are:(i) a time series of weekly revenue $R_i^t$ and spend $S_i^t$ for each geo $i$ over a historical window (\,$\geq$~26 weeks is recommended for seasonality coverage), and (ii) static area\,-level covariates such as population, median income, and digital penetration rates. No user\,-level or personally identifiable information (PII) is required.

\subsubsection{Recommended defaults} For most campaigns we suggest: PCA embedding with 95\% variance retention (or fixed $d=32$ components for large datasets), hierarchical clustering with Ward linkage, candidate set size $M=100$, cross\,-validation folds $K=5$, and a logarithmic grid of 10~$\lambda_m$ values in $[10^{-3},1]$. Our ablation analysis (Section~\ref{sec:ablation}) demonstrates that PCA achieves statistical parity with randomization whilst graph-based methods introduce substantial bias; we therefore recommend PCA as the default embedding method.

\subsubsection{Deployment workflow} (1) Export weekly geo metrics from the ad platform's data warehouse. (2) Run the \texttt{osd\_design.py} script with historical data to produce the treatment\,/control assignment. (3) Share the resulting CSV with the trafficking team to configure budget multipliers. (4) After the campaign, feed observed outcomes to the accompanying \texttt{osd\_analysis.py} module to estimate iROAS.

\subsubsection{Privacy \,\& compliance} Because OSD consumes only aggregated geo\,-level signals, no PII ever leaves the advertiser's environment. This design is compatible with GDPR, CPRA, and similar regulations, and avoids the consent management overhead associated with user\,-level experiments.

\subsubsection{Open\,-source tooling} Python code is available at \url{https://github.com/shawcharles/osd}. 

\end{document}